\documentclass[conference]{IEEEtran}
\IEEEoverridecommandlockouts
\usepackage{cite}
\usepackage{url}
\usepackage{amsmath,amssymb,amsfonts}
\usepackage{amsthm}
\usepackage{algorithmic}
\usepackage{graphicx}
\usepackage{textcomp}
\usepackage{xcolor}
\def\BibTeX{{\rm B\kern-.05em{\sc i\kern-.025em b}\kern-.08em
    T\kern-.1667em\lower.7ex\hbox{E}\kern-.125emX}}

\newcommand\defeq{\mathrel{\stackrel{\makebox[0pt]{\mbox{\normalfont\scriptsize def}}}{:=}}}

\newcommand{\G}{\mathbf{G}}
\newcommand{\F}{\mathbf{F}}

\newcommand{\X}{\mathbf{X}}

\newtheorem{lemma}{Lemma}

\newcommand{\Comment}[1]{}

\begin{document}

\title{Safeguarding Learning-based Control for Smart Energy Systems with Sampling Specifications\\
\thanks{The first two authors contributed equally to this work. This work was funded by the StMWi Bayern to support the thematic development of Fraunhofer IKS. }%\vspace{-3mm}
}

\author{\IEEEauthorblockN{Chih-Hong Cheng\IEEEauthorrefmark{1}, 
Venkatesh Prasad Venkataramanan\IEEEauthorrefmark{1},
Pragya Kirti Gupta\IEEEauthorrefmark{1}, Yun-Fei Hsu\IEEEauthorrefmark{1},  Simon Burton\IEEEauthorrefmark{1} }	\IEEEauthorblockA{\IEEEauthorrefmark{1}Fraunhofer IKS, Munich, Germany}	
}

\maketitle

\begin{abstract}
 
We study challenges using reinforcement learning in controlling energy systems, where apart from performance requirements, one has additional safety requirements such as avoiding blackouts. We detail how these  safety requirements in real-time temporal logic can be strengthened via discretization into linear temporal logic (LTL), such that the satisfaction of the LTL formulae implies the satisfaction of the original safety requirements. The discretization enables advanced engineering methods such as synthesizing shields for safe reinforcement learning as well as formal verification, where for statistical model checking, the probabilistic guarantee acquired by LTL model checking forms a lower bound for the satisfaction of the original real-time safety requirements.   

\end{abstract}


\section{Introduction}

The rapid introduction of renewable energy has led to increased costs to stabilize the electricity network within countries. According to a recent report~\cite{marot2020l2rpn}, stabilization costs exceeded~1.4 billion Euros in Germany in 2017. Blackouts in recent years have adversely affected hospitals~\cite{Healthcare} and forest fires due to overloading of power lines have caused the destruction of habitat and property damage worth millions of dollars~\cite{ForestFire}.  A relatively low-cost approach to stabilizing a network is to use topological changes. However, it is very hard to simulate hundreds of options in real time, meaning that most operators revert to costly production dispatches and load disconnections. 

Stepping away from classical techniques, the use of techniques inspired by the prevalence of deep neural networks has also created great attention. Energy companies such as RTE are investigating whether \emph{reinforcement learning (RL)} can utilize cheap topological changes to the network, thereby improving performance. The fundamental challenge now comes with the problem of ensuring the safety of the energy system with learning-enabled controllers.

In this paper, we demonstrate that for the two primary energy safety requirements on avoiding line overheating and blackout specified as metric temporal logic (MTL)~\cite{koymans1990specifying} formulae~$\phi^{MTL}$, one can create strengthened linear temporal logic (LTL)~\cite{pnueli1977temporal} formulae~$\phi^{LTL, \Delta}$ such that the satisfaction of the LTL formula (where states of the plant are only sampled with a frequency of~$\frac{1}{\Delta}$) implies that the original safety specification  also holds. This altogether enables the legitimate use of techniques as defined in safe reinforcement learning~\cite{gu2022review}, including  shielding~\cite{alshiekh2018safe} or reward shaping for $\omega$-regular objectives~\cite{hahn2020faithful}, where the specification should be characterized using LTL and the training of the RL agent is based on discrete systems. We have applied the method of statistical model checking~\cite{legay2010statistical}, where the immediate corollary is that the probability of satisfaction of the strengthened LTL specification forms a \emph{lower-bound} for the satisfaction of the MTL specification.

\begin{figure}[t]
\centerline{

\includegraphics[width=0.9\columnwidth]{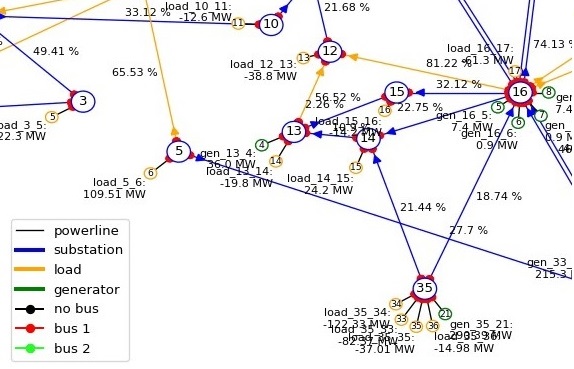}
}
\vspace{-3mm}
\caption{An energy grid model under analysis, highlighting information on production, consumption, and load of every power line.}
\vspace{-5mm}
\label{fig:energy.grid}
\end{figure}

\vspace{-1mm}
\section{Control of an Energy Grid}

For an energy grid $\mathcal{EG}$ similar to the one illustrated in Fig.~\ref{fig:energy.grid}, the overall control diagram can be understood using Fig.~\ref{fig:safety.rl.energy}, where the controller (the RL agent) periodically (with frequency of~$\frac{1}{\Delta}$) reads the following three types of input states (1) the current generation and consumption of nodes in the grid,  (2) the current load of power lines, and (3) the current topology information; it then triggers a control action that may change the topology of the grid. $\mathcal{EG}$ is, by definition, a hybrid system (cf Chap.~4 of~\cite{lee2016introduction}) where time is not discrete but rather continuous, and the change of power network topology implies the change of the system behavior. Given a real-time specification $\phi^{MTL}$ in MTL where the atomic proposition is evaluated on the (observable) state variables of the grid, denote $\mathcal{EG} \models \phi^{MTL}$ if starting from the initial state, the set of all traces satisfies $\phi^{MTL}$.  

Finally, we consider evaluating $\mathcal{EG}$ against an LTL formula~$\phi^{LTL}$ where the atomic proposition is also evaluated on the state variables of the power grid. However, for defining the initial and next state, we only consider the time where the RL agent takes action, i.e., we consider the $\Delta$-sampled infinite state sequence at time
$t = \beta \Delta$ with $\beta \in \mathbb{N}_{0}$. Denote $\mathcal{EG} \models_{\Delta} \phi^{LTL}$ if starting from the initial state, the set of all  $\Delta$-sampled infinite state sequences satisfies~$\phi^{LTL}$.

\section{Towards Sampling-based Specification}

Following the standard notation of LTL and MTL, we use $\G$, $\F$ for characterizing ``globally" and ``eventually". The symbol $\X$ is used exclusively in LTL for ``next", where $\X^j$ is the abbreviation for~$j$ consecutive~$\X$s. 

\paragraph{Specification on no overloaded power lines} For any power line~$i$, let $oload_i$ be the atomic proposition indicating if the power line is overloaded. Then, the specification of ``power line~$i$ should not be overloaded for more than $\kappa$ minutes" can be rephrased using the following MTL formula. 
\begin{equation}
   \phi^{MTL}_{bound.overload} \defeq \G (oload_i \rightarrow \F_{[0, \kappa]} \neg oload_i)
\end{equation}

Consider the system state is checked using sampling time~$\Delta$ where $\kappa > \Delta$. Then we have the following lemma stating that satisfying a strengthened LTL formula $\phi^{LTL,\Delta}_{bound.overload}$ as specified in Eq.~\ref{eq.ltl.overload} ensures  $\phi^{MTL}_{bound.overload}$ to also hold.
\begin{equation}\label{eq.ltl.overload}
   \phi^{LTL,\Delta}_{bound.overload} \defeq \G (oload_i \rightarrow \bigvee^{\lfloor \kappa / \Delta \rfloor - 1}_{j=0} \X^j \neg oload_i)
\end{equation}

\begin{lemma}\label{lemma.overload.discretization.guarantee}
For a energy grid $\mathcal{EG}$, assume that the state is measured under the sampling frequency~$\Delta$, then if $\mathcal{EG} \models_{\Delta} \phi^{LTL,\Delta}_{bound.overload}$, then $\mathcal{EG} \models \phi^{MTL}_{bound.overload}$. 
\end{lemma}

\begin{proof}
(Sketch) We use Fig.~\ref{fig:overload.proof} to assist in understanding the key idea of the proof, where in the timeline, the dashed block represents cases where there is an overload in the considered power line. If the overload is detected at time $t$ but can not be detected at time $t+3\Delta$, the duration of power line overload can at most be $3\Delta + \Delta$. In the general case, if the overload is first detected at $t$ and disappears at~$t + (\lfloor \kappa / \Delta \rfloor -1) \Delta$, the duration of power line overload can only be arbitrarily close to $(\lfloor \kappa / \Delta \rfloor -1) \Delta + \Delta \leq \kappa$. Therefore, if  $\phi^{LTL,\Delta}_{bound.overload}$ holds, the duration of power line overloading can never exceed $\kappa$, implying that $ \phi^{MTL}_{bound.overload}$ also holds.

\end{proof} 

\begin{figure}[t]
\centerline{

\includegraphics[width=\columnwidth]{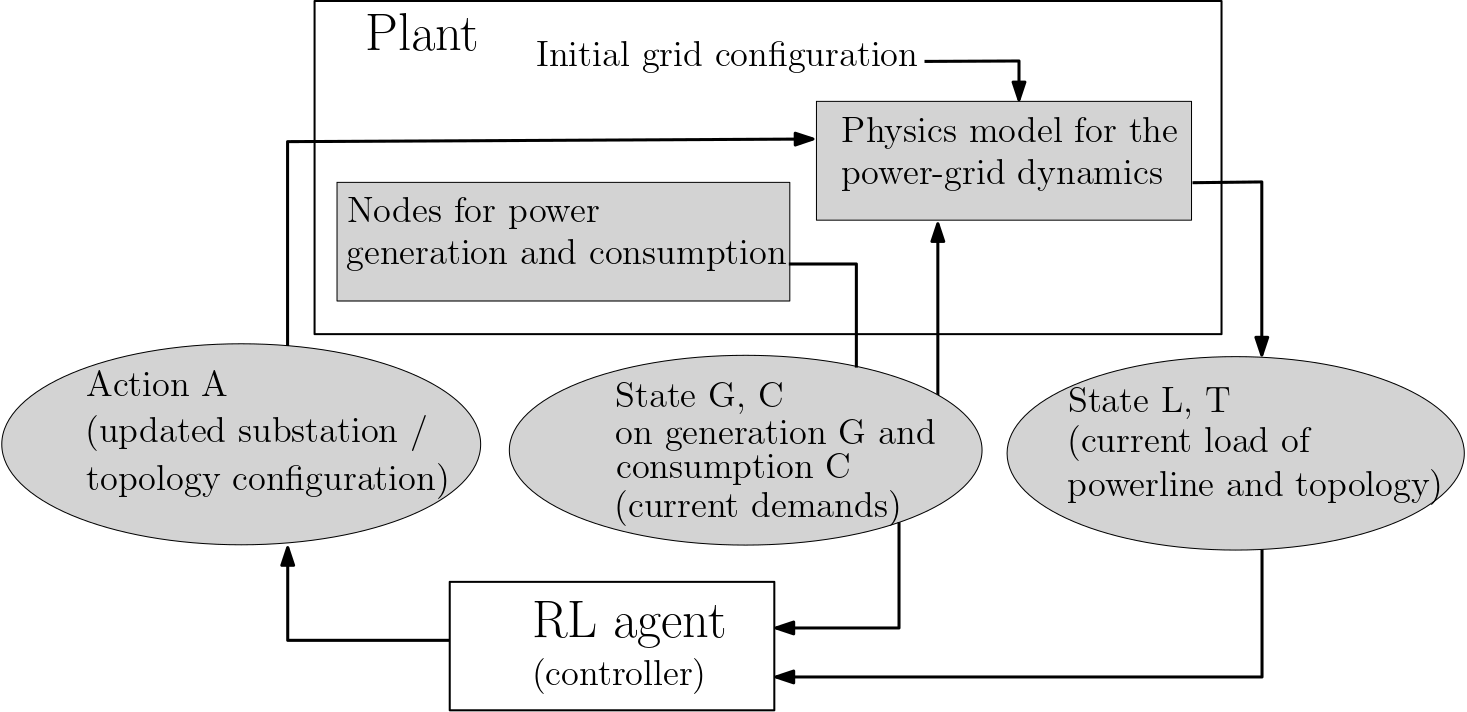}
}
\caption{Illustration on RL for controlling the energy grid.}
\label{fig:safety.rl.energy}
\vspace{-3mm}
\end{figure}

\vspace{-2mm}

\paragraph{Specification on no blackout } For any node~$i$ consuming the energy, let $blackout_i$ be the atomic proposition indicating ``blackout" occurs, i.e.,  a consumer can not receive power while its demand is non-zero. Then the specification of ``consumer~$i$ in the power grid should not have a blackout at any time" can be rephrased using the following MTL formula. 

\begin{equation}
   \phi^{MTL}_{no.blackout} \defeq \G (\neg blackout_i)
\end{equation}

Consider again the system state is checked using sampling time~$\Delta$. Provided that the blackout can not be recovered instantaneously but requires a duration larger than~$\kappa$ (which is a reasonable assumption),  the following lemma stating that satisfying a strengthened LTL formula $\phi^{LTL,\Delta}_{bound.overload}$ as specified in Eq.~\ref{eq.ltl.blackout} ensures the satisfaction of  $\phi^{MTL}_{bound.overload}$.

\vspace{-1mm}

\begin{equation}\label{eq.ltl.blackout}
   \phi^{LTL,\Delta}_{no.blackout} \defeq \G (\neg blackout_i)
\end{equation}

\begin{lemma}\label{lemma.blackout.guarantee}
For an energy grid $\mathcal{EG}$, under the assumption where when a blackout occurs, it takes at least $\gamma$ time to recover where $\gamma > \kappa$, then if $\mathcal{EG} \models_{\Delta}  \phi^{LTL,\Delta}_{no.blackout}$, then $\mathcal{EG} \models \phi^{MTL}_{no.blackout}$. 
\end{lemma}

\begin{proof}
(Sketch) We use Fig.~\ref{fig:blackout.proof} to explain the idea. Intuitively, if any blackout can not be recovered within  $\gamma$ time units, as~$\gamma$ is larger than the sampling period $\Delta$, the evaluation of the LTL formula~$\phi^{LTL,\Delta}_{no.blackout}$ is deemed to be false. Therefore, if~$\phi^{LTL,\Delta}_{no.blackout}$ holds, so does $\phi^{MTL}_{no.blackout}$.

\end{proof}

\begin{figure}[t]
\centerline{

\includegraphics[width=0.8\columnwidth]{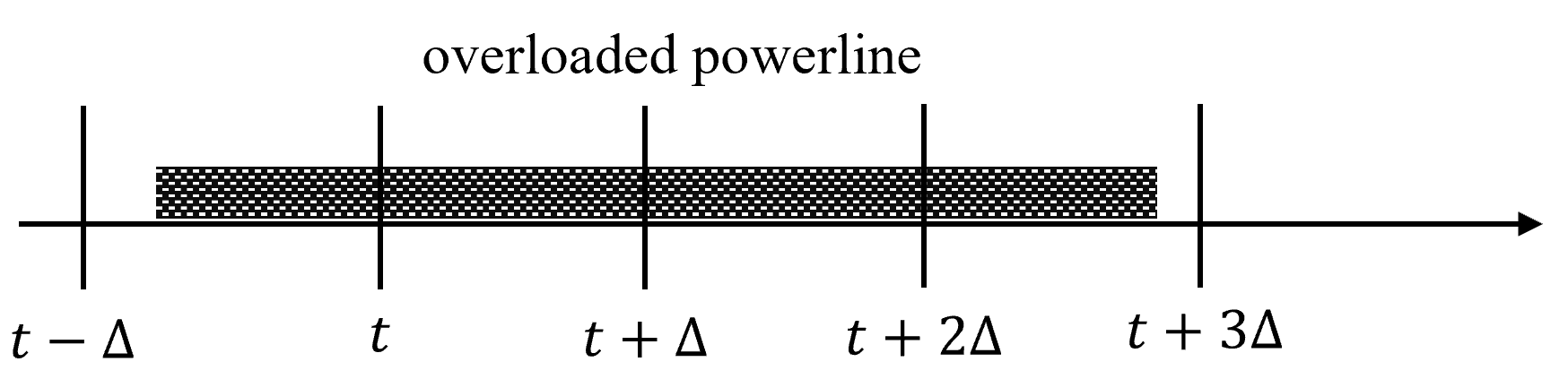}
}
\vspace{-2mm}
\caption{Illustration on the proof strategy for Lemma~\ref{lemma.overload.discretization.guarantee}}
\label{fig:overload.proof}
\vspace{-2mm}
\end{figure}

\begin{figure}[t]
\centerline{

\includegraphics[width=0.8\columnwidth]{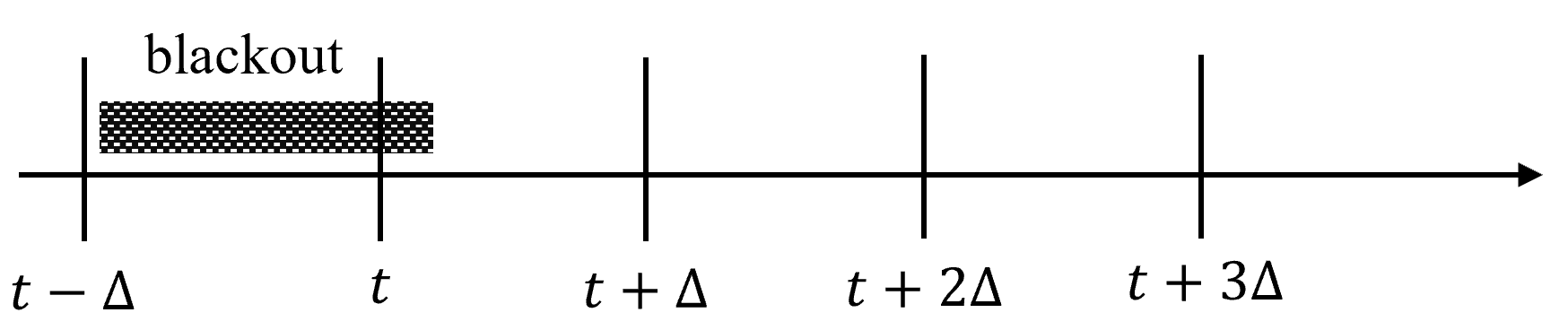}
}
\vspace{-2mm}
\caption{Illustration on the proof strategy for Lemma~\ref{lemma.blackout.guarantee}}
\label{fig:blackout.proof}
\vspace{-5mm}
\end{figure}

\vspace{-3mm}
\section{Application and Future Directions}

As a case study, we have applied statistical model checking to understand the probability of the RL-controlled energy grid model (illustrated in Fig.~\ref{fig:energy.grid}) satisfying the safety specification. The MTL formula set~$\kappa$ to be $10$ minutes, as regulated in~\cite{marot2020l2rpn}. The RL controller is operated under the configuration where~$\Delta$ equals~$5$ minutes. With the production and consumption model made available, the result demonstrated that the RL-controlled energy grid could satisfy the strengthened LTL safety property with a probability of~$0.8912$ under a  target error rate $\alpha = 0.001$. 
Via the proof of Lemma~\ref{lemma.overload.discretization.guarantee}, we know that the probability of satisfying an MTL formula is at least equal to the one of satisfying the translated LTL formula.\footnote{As an example, in Fig.~\ref{fig:overload.proof}, if $\kappa = 4\Delta - \delta$ where $\delta$ is a very small positive constant, then the scenario in Fig.~\ref{fig:overload.proof} will violate the LTL specification while satisfying the MTL specification.} 

The result from this work is our initial step towards the rigorous design of learning-enabled energy systems, where translating the MTL specification to LTL allows us to use established results in safe reinforcement learning. Future research directions include constructing interpretable shields via the sound abstraction of the plant, as well as relaxing assumptions used in analyzing the model.

\bibliographystyle{abbrv}
%\bibliography{ref} 

\end{document}